\documentclass[]{llncs}

\usepackage[T1]{fontenc}
\usepackage[english]{babel}
\usepackage{amssymb,amsmath}
\usepackage{graphicx}
\usepackage{url}

\urldef{\mailsc}\path|imai@iec.hiroshima-u.ac.jp|

\newcommand{\ACA}{\mathcal{A}}
\newcommand{\Conf}{\operatorname{Conf}}
\newcommand{\C}{\mathfrak{C}}
\newcommand{\NN}{\mathbb{N}}
\newcommand{\ZZ}{\mathbb{Z}}
\newcommand{\cyc}[1]{\widetilde{#1}}

\begin{document}
\mainmatter
\title{$5$-State Rotation-Symmetric Number-Conserving Cellular Automata are not Strongly Universal}
\titlerunning{Non-Universality of $5$-State Rotation-Symmetric NCCAs}
\author{Katsunobu Imai\inst{1} \and Hisamichi Ishizaka\inst{1} \and Victor Poupet\inst{2}}
\authorrunning{K. Imai and H. Ishizaka}
\institute{Graduate School of Engineering,
Hiroshima University,\\
Higashi-Hiroshima 739-8527 Japan\\
\mailsc\and
LIRMM, Universit\'e de Montpellier 2\\
161 rue Ada, 34392 Montpellier, France\\
\url{victor.poupet@lirmm.fr}
}
\maketitle

\begin{abstract}
	We study two-dimensional rotation-symmetric number-conserving cellular automata working on the von Neumann neighborhood (RNCA). It is known that such automata with 4 states or less are trivial, so we investigate the possible rules with 5 states. We give a full characterization of these automata and show that they cannot be strongly Turing universal. However, we give example of constructions that allow to embed some boolean circuit elements in a 5-states RNCA.
\end{abstract}

\section{Introduction}

Cellular automata (CA) are widely studied deterministic, discrete and massively parallel dynamical systems. They were introduced by J. von Neumann and S. Ulam in the 1940s (published posthumously in 1966) to study self replicating systems \cite{vonNeumann66}.

Von Neumann developped a 29-states CA working on the ``4 closest cells'' neighborhood (that would later be known as the ``von Neumann neighborhood'') capable of replicating patterns, and used it to describe the notion of computational universality. In 1968, E.~Codd improved von Neumann's construction by devising an 8-states universal rotation-symmetric CA working on the same neighborhood \cite{codd68} and proving that 2-states von Neumann neighborhood CA could not be strongly universal. In 1970, E. Banks proved the existence of a 4-states strongly universal rotation-symmetric CA on von Neumann's neighborhood and a 2-states weakly universal one \cite{banks70}. This search for small universal automata was finally completed in 1987 when T.~Serizawa published a 3-states strongly universal CA on the von Neumann neighborhood \cite{serizawa87}.

Since their introduction, cellular automata have been broadly used as a model for biological, chemical and physical processes. In 1969, K. Zuse was the first to propose the theory that physics is computation and that the universe could be seen as a cellular automaton \cite{zuse69}. The notion of number conservation appeared as a result of attempts to implement physical conservation laws (energy, mass, particles, etc.) in cellular automata. It was first defined by K.~Nagel and M.~Schreckenberg to model road traffic flow \cite{nagel92} and has since become an important subject of research in the cellular automata community as a natural example of global property enforced by local conditions.

In the early 2000s, H.~Fuks and N.~Boccara \cite{boccara02}, B.~Durand et al. \cite{durand03} and A.~Moreira \cite{moreira03} independently proved that number-conservation was a decidable property by giving explicit characterizations of local transition rules of number-conserving cellular automata (NCCA).

Imai et al. proved the existence of a weakly universal 29-states von Neumann neighborhood rotation-symmetric NCCA and this result was later improved by N.~Tanimoto and K.~Imai who managed to reduce the number of required states to 14 by using an exact characterization of von Neumann neighborhood NCCA \cite{tanimoto09}. From the same characterization, N.~Tanimoto et al. proved that von Neumann neighborhood rotation-symmetric NCCA with 4 states or less are all trivial \cite{tanimoto09b}.

In this article, we investigate the computational power of 5-states von Neumann neighborhood rotation-symmetric NCCA. We obtain a precise description of the possible transition rules of such automata and prove that none of them can be strongly universal as their evolution from a finite initial configuration must be ultimately periodic. However, we show that they are capable of some computation by exhibiting patterns that simulate the behavior of some boolean circuit elements.

\section{Definitions}

In this article, we will only consider deterministic 2-dimensional cellular automata working on the von Neumann neighborhood.

\begin{definition}[Cellular automaton]
	A 2-dimensional von Neumann neighborhood cellular automaton (CA) is a triple $\ACA = (Q, f, q_0)$ where
	\begin{itemize}
		\item $Q$ is a finite set. The elements of $Q$ are called \emph{states};
		\item $f:Q^5 \rightarrow Q$ is a function called the \emph{local transition function} of $\ACA$;
		\item $q_0\in Q$ is called the \emph{quiescent state} and $f(q_0,q_0,q_0,q_0,q_0) = q_0$.
	\end{itemize}

	A 2-dimensional \emph{configuration} over $Q$ is a mapping $\C:\ZZ^2\rightarrow Q$. The elements of $\ZZ^2$ are called \emph{cells} and for a cell $c\in\ZZ^2$, we say that $\C(c)$ is the \emph{state} of $c$ in the configuration $\C$.
	
	The set of all configurations over $Q$ is denoted by
	\begin{displaymath}
		\Conf(Q)=\{c, c:\ZZ^2\rightarrow Q\}
	\end{displaymath}
	
	From the local transition function $f$, we define a global transition function $F:\Conf(Q)\rightarrow \Conf(Q)$ obtained by replacing the state of each cell by the result of $f$ applied to the cell's state and the states of its $4$ closest neighbors\footnote{This is the von Neumann neighborhood.}:\\
	$\forall \C\in\Conf(Q),$
	\begin{displaymath}
		F(\C) = \left\{\begin{array}{rl}
			\ZZ^2 &\rightarrow Q \\
			(x,y) &\mapsto f(\C(x,y), \C(x, y+1), \C(x+1, y), \C(x, y-1), \C(x-1, y))
			\end{array}\right.
	\end{displaymath}
\end{definition}

As is commonly done in dynamical systems theory, we will use the same notation for the cellular automaton and its global transition function. The image of a configuration $\C$ by the global transition function of an automaton $\ACA$ will therefore be denoted $\ACA(\C)$.

\begin{definition}[Finite configurations]
	
	Given a cellular automaton $\ACA=(Q, f, q_0)$, we say that a configuration $\C\in \Conf(Q)$ is \emph{finite} if only a finite number of cells are in a state other than $q_0$ in $\C$.
	
	The \emph{rectangular bound} of a finite configuration $\C$ is the smallest rectangle containing all the non-quiescent cells of $\C$. We call \emph{width} and \emph{height} of the configuration $\C$ the width and height of its rectangular bound.
\end{definition}

From the locality of the transition rule of the automaton and the fact that $f(q_0,q_0,q_0,q_0,q_0) = q_0$ it is clear that the image by a CA of a finite configuration $\C$ of dimensions $(w\times h)$ is also finite and of dimensions $(w'\times h')$ with $w'\leq w+2$ and $h'\leq h+2$.

\begin{definition}[Rotation-symmetry]
	
	A von Neumann neighborhood cellular automaton $\ACA=(Q,f,q_0)$ is said to be \emph{rotation-symmetric} if its local transition function $f$ is invariant by rotation of a quarter cycle:\\
	\begin{displaymath}
		\forall c, u, r, d, l \in Q,\quad
		f(c,u,r,d,l) = f(c,r,d,l,u)
	\end{displaymath}
\end{definition}

If the set of states of the automaton is a subset of $\ZZ$, it is possible to quantify the variations of state values between a configuration and its image. The notion of number-conservation is inspired from physical conservation laws and states that the total sum of all states on a configuration is conserved by the global transition function.

\begin{definition}[Number conservation]
	
	A \emph{number-conserving cellular automaton} (NCCA) is a cellular automaton $\ACA = (Q, f, q_0)$ such that $Q\subset \ZZ$ and for any finite configuration $\C\in\Conf(Q)$,
	\begin{displaymath}
		\sum_{c\in\ZZ^2} \ACA(\C)(c) - \C(c) = 0
	\end{displaymath}
\end{definition}

There are other commonly used definitions of number-conservation, some of which apply to infinite configurations as well, but most of these definitions are known to be equivalent~\cite{durand03}.

In this article, we will be studying von Neumann neighborhood rotation-symmetric number-conserving cellular automata. These will be denoted RNCA from now on.

\begin{definition}[Trivial automaton]
	A cellular automaton $\ACA=(Q,f,q_0)$ is said to be \emph{trivial} if for every configuration $\C\in\Conf(Q)$, $\ACA(\C) = \C$.
\end{definition}

\begin{definition}[Strong Turing universality]
	A cellular automaton $\ACA$ is said to be strongly Turing-universal if it can simulate any deterministic Turing machine from a finite initial configuration.
\end{definition}

\section{Characterization of 5-states RNCA}
\label{sec:characterization}

In this section, we consider 5-states RNCA. We use the characterization by Tanimoto et al. \cite{tanimoto09} to show that all RNCA with less than 5 states are trivial and describe precisely the possible local transition functions of 5-states RNCA.

\begin{theorem}[Tanimoto et al. \cite{tanimoto09}]
	\label{theo:rsncca}
	A rotation-symmetric von Neumann neighborhood CA $\ACA=(Q, f, q_0)$ with $Q\subset \ZZ$ is number-conserving if and only if\\
	$\exists g,h: Q^2 \rightarrow \mathbb{Z},\; \forall c,u,r,d,l \in Q,$
\begin{align*}
	f(c,u,r,d,l) = &\; c + g(c,u)  + g(c,r)  + g(c,d)  + g(c,l) \\
		& + h(u,r) + h(r,d) + h(d,l) + h(l,u) \\
	g(c,u) = & -g(u,c),\qquad h(u,r) = -h(r,u)
\end{align*}
\end{theorem}

The function $g$ represents direct transfers of value to the central cell from its neighbors (along the horizontal and vertical directions). The function $h$ corresponds to indirect (diagonal) transfer between the neighbors of the central cell. The functions $g$ and $h$ are called \emph{flow functions} of the automaton $\ACA$.

\textbf{Remark:} If $\ACA=(Q, f, q_0)$ is a RNCA, the function $g$ is uniquely defined by $f$ since for all states $x, y\in Q$,
\begin{displaymath}
	f(x, y, y, y, y) = x + 4g(x, y) + 4h(y, y) = x + 4g(x, y)
\end{displaymath}
As for the function $h$, for all states $x, y, z, t\in Q$ the value of $h(x, y) + h(y, z) + h(z, t) + h(t, x)$ is uniquely defined but there are multiple functions matching this condition as discussed in Subsection~\ref{ssec:indirect_flow}.

\subsection{Direct Flow}

\begin{lemma}
	\label{lem:gtrivial}
	For a RNCA $\ACA=(Q,f,q_0)$ with flow functions $g$ and $h$, if $g\equiv 0$ then $\ACA$ is trivial.
\end{lemma}

\begin{proof}
	We show that if the automaton is not trivial, it must have infinitely many states. According to Theorem~\ref{theo:rsncca}, if $g\equiv 0$, the local transition function of the automaton is
	\begin{displaymath}
		f(c,u,r,d,l)
		= c + h(u, r) + h(r, d) + h(d, l) + h(l, u)
	\end{displaymath}
	
	If $\ACA$ is not trivial, there exist states $c, u, r, d, l\in Q$ and an integer $\delta \neq 0$ such that
	\begin{displaymath}
		f(c,u,r,d,l) = c + \delta
	\end{displaymath}
	so $h(u, r) + h(r, d) + h(d, l) + h(l, u)=\delta$ and for any state $q\in Q$,
	\begin{displaymath}
		f(q,u,r,d,l)
		= q + h(u, r) + h(r, d) + h(d, l) + h(l, u)
		= q + \delta
	\end{displaymath}
	This means that for any $q\in Q$, $(q+\delta)$ is also a state of $\ACA$ which is not possible if $Q$ is finite.\qed
\end{proof}

\begin{lemma}
	\label{lem:minstates}
	Let $\ACA = (Q, f, q_0)$ be a RNCA with flow functions $g$ and $h$. For any two states $a,b\in Q$ if we denote $g(a,b)=\alpha$, then
	\begin{displaymath}
		\{a, a + \alpha, a + 2\alpha, a + 3\alpha, a + 4\alpha, b, b - \alpha, b - 2\alpha, b - 3\alpha, b - 4\alpha\} \subseteq Q
	\end{displaymath}
\end{lemma}
\begin{proof}
	We know that
	\begin{align*}
		g(a,b) &= -g(b,a) = \alpha, &
		g(a,a) &= g(b,b) = 0,\\
		h(a,b) &= -h(b,a), &
		h(a,a) &= h(b,b) = 0
	\end{align*}
	
	By Theorem~\ref{theo:rsncca}, we have
	\begin{align*}
		f(a,a,a,a,b) &= a + 3g(a,a) + g(a,b) + h(a,b) + h(b,a) + 2h(a,a)\\
		&= a + g(a,b) = a + \alpha
	\end{align*}
	Similarly, we get
	\begin{align*}
		f(a,a,a,b,b) &= a + 2\alpha, &
		f(a,a,b,b,b) &= a + 3\alpha, &
		f(a,b,b,b,b) &= a + 4\alpha, \\
		f(b,b,b,b,a) &= b - \alpha, &
		f(b,b,b,a,a) &= b - 2\alpha, &
		f(b,b,a,a,a) &= b - 3\alpha, \\
		f(b,a,a,a,a) &= b - 4\alpha
	\end{align*}
	\qed
\end{proof}

As a consequence of Lemmas~\ref{lem:gtrivial} and \ref{lem:minstates}, if a RNCA $\ACA=(Q, f, q_0)$ with flow functions $g$ and $h$ is not trivial, there exist two states $a,b\in Q$ such that $g(a, b) = \alpha > 0$ and the set of states of the automaton contains 5 elements in arithmetic progression from $a$ of difference $\alpha$. Since the 5 first elements of the arithmetic progression from $b$ of difference $-\alpha$ are also in $Q$, the only way for the CA to have only 5 states is that
\begin{displaymath}
	\begin{array}{l}
	b = a + 4\alpha \\
	\textrm{and}\quad\left\{\begin{array}{rl}
		g(a, b) &= \alpha \\
		g(b, a) &= -\alpha \\
		g(x, y) &= 0 \quad \textrm{otherwise}\\
	\end{array}\right.
	\end{array}
\end{displaymath}

\subsection{Indirect Flow}
\label{ssec:indirect_flow}

Let us now consider the possibilities for $h$. The function $h$ is slightly more complex than $g$ because it is only properly characterized on cycles: the exact value of $h(a,b)$ for two states has no meaning in terms of CA dynamics, the real constraints are on the sums $h(a,b)+h(b,c)+h(c,d)+h(d,a)$ for states $a,b,c,d$.

\begin{definition}
	Given a finite set $Q\subseteq \ZZ$ and a function $h:Q^2\rightarrow \ZZ$ such that for all $a,b\in Q$, $h(a,b)=-h(b,a)$, we define the \emph{cyclic extension} $\cyc h$ of $h$ as
	\begin{displaymath}
		\cyc h : \left\{\begin{array}{rcl}
			Q^* & \rightarrow & \ZZ \\
			(q_1,q_2,\ldots, q_n) & \mapsto & h(q_1,q_2)+h(q_2,q_3)+\ldots+h(q_{n-1}, q_n) + h(q_n, q_1)
			\end{array}\right.
	\end{displaymath}
\end{definition}

Note that although $\cyc h$ is entirely defined by $h$, different functions can have the same cyclic extension. For instance if $h'(x,y) = y-x+h(x,y)$ then $\cyc h = \cyc{h'}$. However, in Theorem~\ref{theo:rsncca} only the cyclic extension of $h$ is actually used to define the behavior of the automaton.

The function $\cyc h$ is entirely defined by its values on triples:\\
$\forall q_1, q_2, \ldots, q_n \in Q,$
\begin{align*}
	\cyc h(q_1, q_2, \ldots, q_n) =\; &h(q_1, q_2)+\ldots+h(q_n, q_1) \\
	=\; &h(q_1,q_2) + h(q_2,q_3) + h(q_3,q_1) \\
	&+ h(q_1,q_3) + h(q_3,q_4) + h(q_4,q_1) \\
	&+ \ldots \\
	&+ h(q_1,q_{n-1}) + h(q_{n-1},q_n) + h(q_n, 1) \\
	=\; &\cyc h(q_1,q_2,q_3) + \cyc h(q_1,q_3,q_4) + \ldots + \cyc h(q_1,q_{n-1},q_n)
\end{align*}

Moreover, it is easy to check from the definition that $\cyc h$ is
\begin{itemize}
	\item null on pairs: $\cyc h(a,b) = 0$;
	\item invariant by repetition of a state: $\cyc h(a,a,q_1,q_2,\ldots, q_n) = \cyc h(a,q_1,q_2,\ldots, q_n)$;
	\item invariant by rotation: $\cyc h(q_1,q_2,\ldots, q_n) = \cyc h(q_2,q_3,\ldots, q_n,q_1)$;
	\item anti-symmetric: $\cyc h(q_1,q_2,\ldots, q_n)=-\cyc h(q_n,q_{n-1},\ldots, q_1)$.
\end{itemize}

Let us now consider a 5-states non trivial RNCA $\ACA = (Q, f, q_0)$ with flow functions $g$ and $h$. We have already shown that the states of the automaton are $Q=\{a+i\alpha\ |\ 0\leq i\leq4\}$. Without loss of generality we can renormalize the states to $Q=\{0,1,2,3,4\}$ (substracting a constant to all states and dividing by a common factor does not affect the number conservation of the CA). We also know that $g(0,4)=1$, $g(4, 0)=-1$ and for all other $x,y\in Q$, $g(x,y)=0$.

First, let us consider $x,y,z\in Q\setminus \{0\}$. From Theorem~\ref{theo:rsncca}, we have
\begin{align*}
	f(4,x,x,y,z) &= 4 + 2g(4,x) + g(4,y) + g(4,z) + \cyc h(x,x,y,z)\\
	&= 4 + \cyc h(x,y,z)
\end{align*}
This means that $(4 + \cyc h(x,y,z))\in Q$ and since $4$ is the largest element of $Q$, we have $\cyc h(x,y,z)\leq 0$. Because $x$, $y$ and $z$ could be any state other than 0, the same reasoning gives $\cyc h(z,y,x) = -\cyc h(x,y,z) \leq 0$ which implies $\cyc h(x,y,z) = 0$.

Similarly, for states $x,y,z\in Q\setminus \{4\}$, by considering $f(0,x,x,y,z)$ we get $\cyc h(x,y,z)=0$.

So the only possible non-zero triples of $\cyc h$ are triples containing both $0$ and $4$. Because $\cyc h(0,0,4) = \cyc h(0,4) = 0$, and $\cyc h(4,4,0) = \cyc h(4,0) = 0$, the only possible non-zero triples of $\cyc h$ are those including $0$, $4$ and a third different state.

Moreover, for $x,y\in Q\setminus\{0,4\}$, we have
\begin{align*}
	\cyc h(0,4,x,y) &= h(0,4) + h(4,x) + h(x,y) + h(y,0)\\
	&= h(0,4) + h(4,x) + h(x, 0) + h(0,x) + h(x,y) + h(y,0)\\
	&= \cyc h(0,4,x) + \cyc h(0,x,y) \\
	&= \cyc h(0,4,x)
\end{align*}
but also
\begin{align*}
	\cyc h(0,4,x,y) &= \cyc h(4,x,y,0) \\
	&= h(4,x) + h(x,y) + h(y,0) + h(0,4)\\
	&= h(4,x) + h(x,y) + h(y,4) + h(4,y) + h(y,0) + h(0,4)\\
	&= \cyc h(4,x,y) + \cyc h(4,y,0) \\
	&= \cyc h(4,y,0) = \cyc h(0,4,y)
\end{align*}

So for $x,y\in Q\setminus\{0,4\}$, $\cyc h(0,4,x) = \cyc h(0,4,y)$. All that remains to do now is consider what are the possible values for $\cyc h(0,4,1)$. By considering
\begin{align*}
	f(0,0,0,4,1) &= 0 + g(0,4) + \cyc h(0,0,4,1) \\
	&= 1 + \cyc h(0,4,1)
\end{align*}
we get that $(1+\cyc h(0,4,1))\in Q$ so $\cyc h(0,4,1) \geq -1$. Similarly, by considering $f(4,0,4,4,1)$ we get $\cyc h(0,4,1) \leq 1$.

The value of $\cyc h(0,4,1)$ fully defines $\cyc h$, as we know that the function is zero-valued on all triples not containing $0$ and $4$ and all values on triples containing $0$ and $4$ can be obtained from $\cyc h(0,4,1)$ by cycle, symmetry and replacing 1 with any other state in $Q\setminus\{0,4\}$. A possible function $h$ that would correspond to such an $\cyc h$ would be defined by
\begin{align*}
	h(0,4) &= \cyc h(0,4,1)\\
	h(4,0) &= -\cyc h(0,4,1)\\
	h(x,y) &= 0 \quad \textrm{otherwise}
\end{align*}

We have therefore proved the following result
\begin{lemma}
	\label{lem:5stchar}
	If $\ACA$ is a 5-states non-trivial RNCA, there exists a constant $\beta\in\{-1,0,1\}$ such that $\ACA$ is equivalent (up to state renaming) to an RNCA with states $Q = \{0,1,2,3,4\}$ and flow functions $g$ and $h$ such that:
	\begin{align*}
		\begin{split}
			g(0,4) &= 1\\
			g(4,0) &= -1\\
			g(x, y) &= 0 \quad \textrm{otherwise}
		\end{split}
		\begin{split}
			h(0, 4) &= \beta \\
			h(4, 0) &= -\beta \\
			h(x, y) &= 0 \quad \textrm{otherwise}
		\end{split}
	\end{align*}
\end{lemma}

\textbf{Remark: } Choosing $\beta = 0$ corresponds to $h\equiv 0$, which leads to a very simple CA for which states 1, 2 and 3 are permanent (a cell in one of these states can never change) and the state of a cell in state 0 (resp. 4) increases (resp. decreases) by the number of its neighbors in state 4 (resp. 0).

Choosing $\beta = 1$ or $\beta = -1$ leads to two different RNCA that are mirror images of each other.

\section{The Power of 5-states RNCA}

In this section we use the characterization of 5-states RNCA from Lemma~\ref{lem:5stchar} to investigate the computational power of such automata. We show that although they cannot be strongly Turing universal, they can simulate some logical circuit elements, indicating that they can perform some sorts of computations.

\subsection{Strong Universality}

\begin{theorem}
	\label{theo:periodic}
	The evolution of a 5-states RNCA $\ACA$ from a finite configuration $\C$ is ultimately periodic:
	\begin{displaymath}
		\exists i, j\in \NN, \quad i<j, \quad \ACA^i(\C) = \ACA^j(\C)
	\end{displaymath}
\end{theorem}

\begin{proof}
	Consider a non-trivial 5-states RNCA $\ACA$. From Lemma~\ref{lem:5stchar} we can assume that $\ACA = (Q, f, q_0)$ with $Q=\{0,1,2,3,4\}$ and consider that its flow functions $g$ and $h$ satisfy the descriptions of the lemma. We show that starting from a finite configuration $\C_0$, non-quiescent states cannot appear arbitrarily far in any direction.
	
	There are two cases to consider, depending on the choice of the quiescent state $q_0$. First if $q_0\in Q\setminus\{0,4\}$, we show that non-quiescent states cannot appear outside of the bounding rectangle of the starting configuration $\C_0$ (see Figure~\ref{fig:middlequiescent}).
	
	\begin{figure}[htbp]
		\centering
			\includegraphics[scale=1.5]{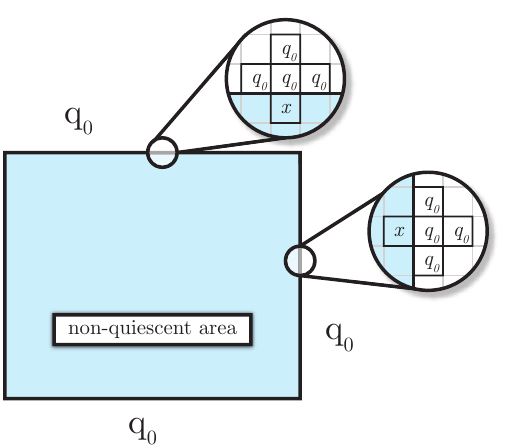}
		\caption{If $q_0\in Q\setminus\{0 ,4\}$, cells outside of the rectangular bounds of the initial configuration (represented in blue) remain in state $q_0$ at all times.}
		\label{fig:middlequiescent}
	\end{figure}
	
	By induction, if at time $t$ all cells outside of the rectangular bounds of $\C_0$ are quiescent, then each of these cells has at least 3 neighbors in state $q_0$. Since $q_0\neq 0$ and $q_0\neq 4$, for all $x\in Q$, $g(q_0,x) = \cyc h(q_0,q_0,q_0,x) = 0$, so $f(q_0, q_0, q_0, q_0, x) = q_0$ and all cells outside of the rectangular bound of $\C_0$ remain in state $q_0$ at time $(t+1)$.
	
	The second case is when $q_0\in\{0,4\}$. Assume that $q_0=0$ and let $N_0$ be the total sum of all the states in $\C_0$ (the argument for $q_0=4$ is similar but we consider the sum of $(q_0-q)$ for all states $q$ in $\C_0$). Since $\ACA$ is number conserving, the total sum of all states in subsequent configurations of the automaton remains equal to $N_0$. Moreover, because all non-quiescent states are positive, the total number of non-quiescent states on any configuration generated from $\C_0$ cannot exceed $N_0$.
	
	If quiescent states appear arbitrarily far in one direction in the evolution of $\ACA$ from $\C_0$, some cells must go from a non-quiescent state back to the quiescent state $0$ in order to keep the total number of non-quiescent states bounded. We will now assume that the direction in which the non-quiescent states appear arbitrarily far is left (it is enough to consider one direction since the automaton is rotation-symmetric).

	Assume that $h(0,4)\geq 0$ and consider the situation illustrated by part (a) of Figure~\ref{fig:corners}: a cell in a non-quiescent state $x$ with 4 cells in state 0 around it, 3 on its right side and one under (for the case $h(0,4)\leq 0$ we would consider the symmetric situation with a quiescent cell on top of the cell in state $x$ instead of under it). Let us see how such a cell can change to state $0$. There are two possibilities, represented by parts (b) and (c) in Figure~\ref{fig:corners}. Either $x=4$ and we can use the function $g$ to lower the state of the cell to $0$ or the contribution of $g$ on the cell is 0 and only $\cyc h$ can lower the state of the cell, in which case only state 1 can be lowered to $0$.
	
	\begin{figure}[htbp]
		\centering
			\includegraphics[scale=1]{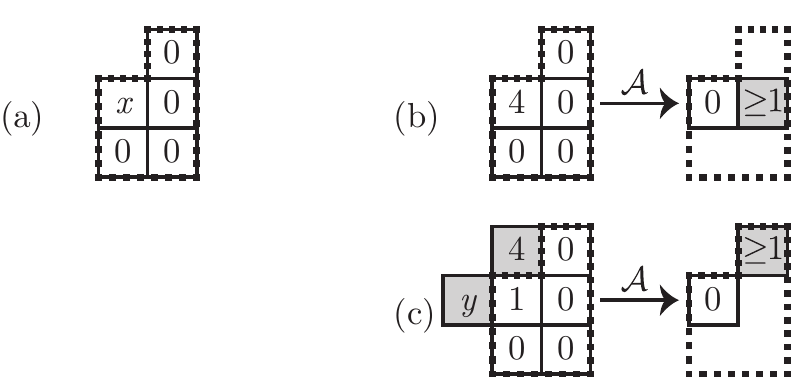}
		\caption{Case study. If a cell is in the situation (a) where $x$ is a non-quiescent state, in order to change to state 0 one of the 3 cells on the right must change from state 0 to a non-quiescent state.}
		\label{fig:corners}
	\end{figure}
		
	(b) If $x=4$, then the cell on the right of the cell in state $4$ will necessarily change to a positive state. For this cell, the contribution from $g$ will be positive (at least 1 from the central cell) and the contribution from $h$ will be 0 since $\cyc h(0,4,0,y)=0$ for all states $y$.
	
	(c) If $x=1$, then in order to change the state to $0$, the contribution from $\cyc h$ must be -1 since $g(1,y)=0$ for all states $y$. We have previously chosen to consider the case $h(0,4)\geq 0$. If $h(0,4)=0$, then the contribution from $h$ is 0 and the state of the cell cannot change to 0. However if $h(0,4)=1$, then the contribution from $\cyc h$ can be $-1$ only if the top neighbor is in state 4 and the left neighbor is in a state $y\in Q\setminus\{0,4\}$ (as illustrated by Figure\ref{fig:corners}). In that case, $\cyc h(4,0,0,y)=-1$ and the cell changes to state $0$. However, the cell on the top-right now has a neighbor in state 4, which means that for this cell the contribution from $g$ is at least 1. Moreover, the contribution from $\cyc h$ is at least 0, since $h(0,4)=1$ and so $\cyc h(0,4,x,y)\geq 0$ for all states $x,y$. This means that the cell on the top-right of the considered cell changes to a positive state.
	
	We have shown that if a cell in the situation illustrated by part (a) of Figure~\ref{fig:corners} changes its state to 0, then at least one of the cells from the column right of the considered cell changes from state 0 to a non-quiescent state. This means that as new non-quiescent states appear towards the left, previous non-quiescent states cannot be properly removed: in order to remove all non-quiescent states from a given column, it is necessary to create new ones on the column at its right, so as the non-quiescent states move towards the left, new ones appear towards the right. Eventually, the total number of non-quiescent states will be greater than $N_0$ which contradicts number-conservation.
	
	So we know that the evolution $(\ACA^i(\C_0))_{i\in\NN}$ of a 5-states RNCA $\ACA$ from a finite configuration $\C_0$ can only have non-quiescent states inside of a bounded area. Because there are only finitely many such bounded configurations, eventually the automaton must re-enter a previous configuration.\qed
\end{proof}

\begin{corollary}
	There are no 5-states strongly Turing universal RNCA.
\end{corollary}
\begin{proof}
	The evolution of a 5-states RNCA from a finite starting configuration is ultimately periodic (Theorem~\ref{theo:periodic}) and therefore decidable. If such CA could simulate a Turing machine then the halting problem and all other behavioral problems on Turing machines would be decidable.\qed
\end{proof}

\subsection{Simulation of Logical Circuits}

Although the evolution of a 5-states RNCA from a finite configuration is ultimately periodic, some non-trivial behaviors can still be observed. In particular, it is possible to simulate some key elements of boolean circuits. In this section, we consider the 5-states RNCA obtained by choosing $\beta=1$ in the characterization from Lemma~\ref{lem:5stchar}.

Figures~\ref{fig:widgetWire}, \ref{fig:widgetBranch} and \ref{fig:widgetSpeed} illustrate how to create a simple wire along which a signal can travel. The wire is made of two layers of state 2 (red), and the signal is represented by two cells in state 4 (yellow) that are ``pushed'' forward by a cell in state 1 (blue). As the signal traverses the wire, the top layer of red states is changed into blue and green states. These wires can therefore be traversed only once, which is enough for simple boolean circuit simulation but not for being used as a control circuit for a universal machine.

\begin{figure}[htbp]
	\centering
		\includegraphics[page=1,scale=1]{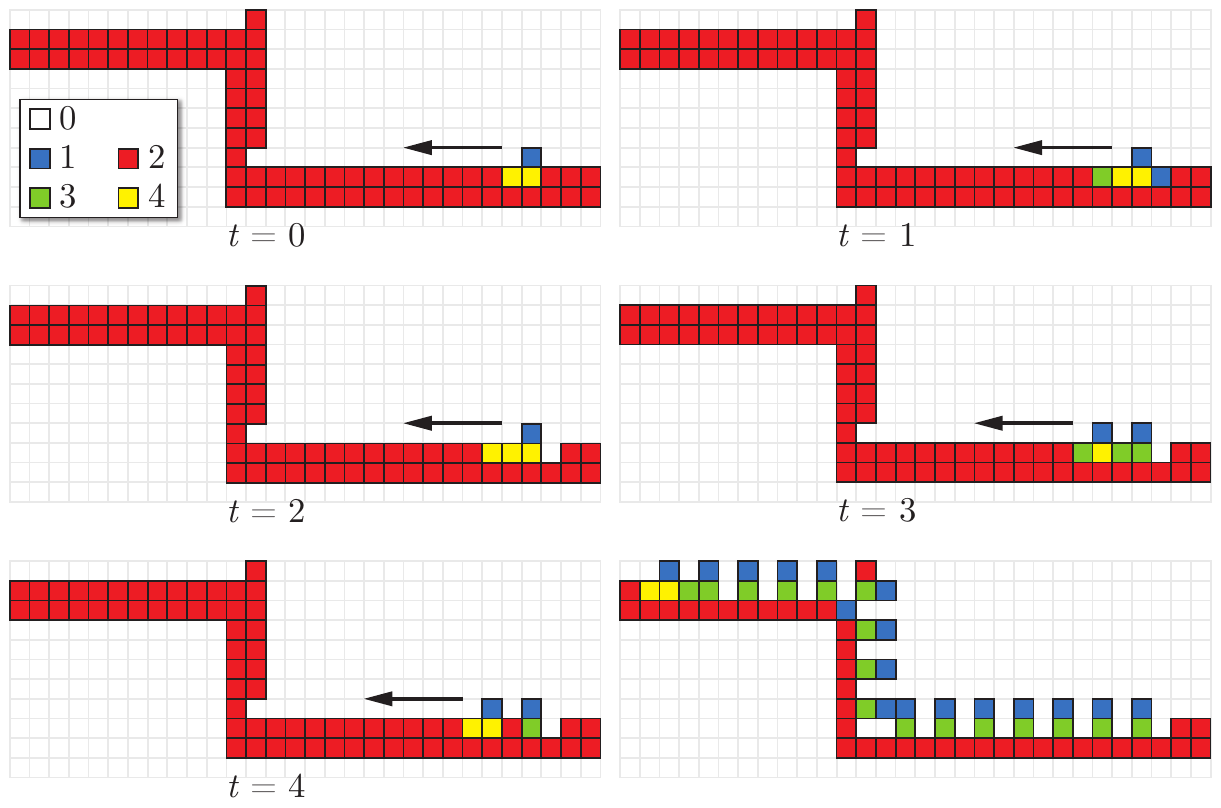} 
	\caption{A signal moving through a simple wire.}
	\label{fig:widgetWire}
\end{figure}

\begin{figure}[htbp]
	\centering
		\includegraphics[page=2,scale=1]{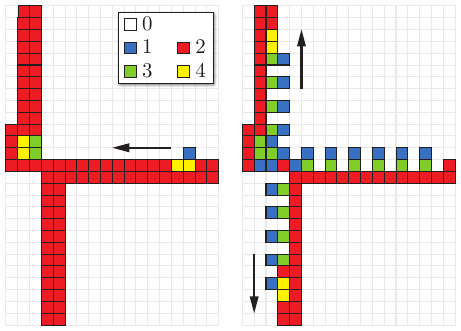}
	\caption{A branching wire.}
	\label{fig:widgetBranch}
\end{figure}

\begin{figure}[htbp]
	\centering
		\includegraphics[page=5,scale=1]{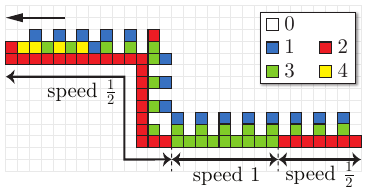}
	\caption{By using state $3$ (green) instead of state $2$ (red) on some portions of the wire, the signal moves at speed $1$ instead of speed $\frac{1}{2}$.}
	\label{fig:widgetSpeed}
\end{figure}

The wires can be split in two to duplicate a signal (Figure~\ref{fig:widgetBranch}). By changing the second row of red cell into green states, the signal can be accelerated to move by one cell at each time step (Figure~\ref{fig:widgetSpeed}) which can be a convenient way to synchronize signals.

As for logical gates, Figures~\ref{fig:widgetAND} and \ref{fig:widgetNAND} illustrate how to implement an \texttt{AND} and a    \texttt{$\overline{A}$ AND B} gate. It is known that the \texttt{$\overline{A}$ AND B} gate alone is sufficient to simulate a \texttt{NOT} and an \texttt{OR} gate so these elements would be sufficient for boolean operations \cite{banks70}.

\begin{figure}[htbp]
	\centering
		\includegraphics[page=3,scale=1]{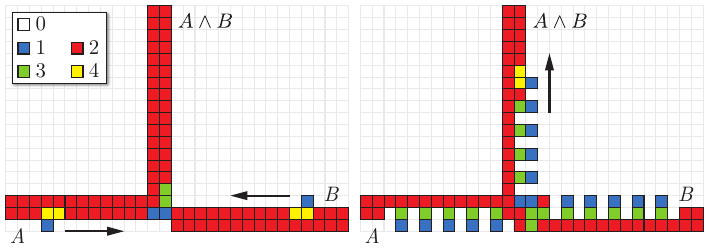}
	\caption{Logical $\operatorname{AND}$ gate.}
	\label{fig:widgetAND}
\end{figure}

\begin{figure}[htbp]
	\centering
		\includegraphics[page=4,scale=1]{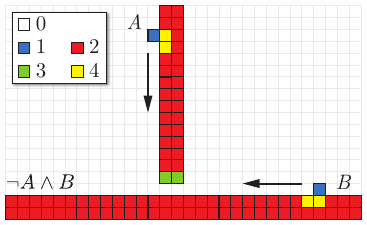}
	\caption{$\overline{A} \operatorname{AND} B$ gate.}
	\label{fig:widgetNAND}
\end{figure}

Note that the boolean value for 0 is simply represented by the absence of a signal. Therefore, careful synchronisation is required when implementing multiple input logical gates as the possible input signals must arrive at the same time (but this can be done either by using the variable speed wires from Figure~\ref{fig:widgetSpeed} or by artificially increasing the length of a wire with a detour.

However, until now we have been unable to devise a wire crossing pattern which is required for a full simulation of boolean circuits. Because of the destructive nature of signal propagation along a wire, wire crossing might prove impossible to implement.


\begin{thebibliography}{10}

\bibitem{banks70}
Edwin~Roger Banks.
\newblock Universality in cellular automata.
\newblock In {\em SWAT (FOCS)}, pages 194--215. IEEE Computer Society, 1970.

\bibitem{boccara02}
Nino Boccara and Henryk Fuk\'s.
\newblock Number-conserving cellular automaton rules.
\newblock {\em Fundam. Inform.}, 52(1-3):1--13, 2002.

\bibitem{codd68}
Edger~Frank Codd.
\newblock {\em Cellular automata}.
\newblock ACM monograph series. Academic Press, 1968.

\bibitem{durand03}
Bruno Durand, Enrico Formenti, and Zsuzsanna R{\'o}ka.
\newblock Number-conserving cellular automata i: decidability.
\newblock {\em Theor. Comput. Sci.}, 1-3(299):523--535, 2003.

\bibitem{moreira03}
Andr\'es Moreira.
\newblock Universality and decidability of number-conserving cellular automata.
\newblock {\em Theoretical Computer Science}, 292(3):711 -- 721, 2003.
\newblock Algorithms in Quantum Information Prcoessing.

\bibitem{nagel92}
Kai Nagel and Michael Schreckenberg.
\newblock {A cellular automaton model for freeway traffic}.
\newblock {\em Journal de Physique I}, 2(12):2221--2229, December 1992.

\bibitem{serizawa87}
Teruo Serizawa.
\newblock Three-state neumann neighbor cellular automata capable of
  constructing self-reproducing machines.
\newblock {\em Systems and Computers in Japan}, 18(4):33--40, 1987.

\bibitem{tanimoto09}
Naonori Tanimoto and Katsunobu Imai.
\newblock A characterization of von neumann neighbor number-conserving cellular
  automata.
\newblock {\em J. Cellular Automata}, 4(1):39--54, 2009.

\bibitem{tanimoto09b}
Naonori Tanimoto, Katsunobu Imai, Chuzo Iwamoto, and Kenichi Morita.
\newblock On the non-existance of rotation-symmetric von neumann neighbor
  number-conserving cellular automata of which the state number is less than
  four.
\newblock {\em IEICE Transactions}, 92-D(2):255--257, 2009.

\bibitem{vonNeumann66}
John {von Neumann}.
\newblock {\em Theory of Self-Reproducing Automata}.
\newblock University of Illinois Press, Urbana, IL, USA, 1966.

\bibitem{zuse69}
K.~Zuse.
\newblock {\em Rechnender {Raum}}.
\newblock Friedrich Vieweg \& Sohn, Braunschweig, 1969.

\end{thebibliography}

\end{document}